\documentclass[11pt,amssymb]{article}

\usepackage{amssymb}
\usepackage{epsfig}
\usepackage{fullpage}

\newcommand{\comment}[1]{}

\newtheorem{theorem}{Theorem}
\newtheorem{definition}{Definition}
\newtheorem{lemma}[theorem]{Lemma}

\newtheorem{corollary}{Corollary}

\def\squareforqed{\hbox{\rlap{$\sqcap$}$\sqcup$}}

\newtheorem{example}{Example}

\newcommand{\qed}{\vrule height6pt width4pt\medskip}

\newcommand{\QED}{\hfill\qed}
\def\squareforqed{\hbox{\rlap{$\sqcap$}$\sqcup$}}
\def\qed{\ifmmode\squareforqed\else{\unskip\nobreak\hfil
\def\mod{\ {\rm mod}\ }
\penalty50\hskip1em\null\nobreak\hfil\squareforqed
\parfillskip=0pt\finalhyphendemerits=0\endgraf}\fi}

\newcommand{\F}{{\Bbb F}}
\newcommand{\N}{{\mathbb N}}

\begin{document}
\title{Improved Deterministic Length Reduction}

\author{
\begin{tabular}{ccc}
Amihood Amir\thanks{ Department of Computer Science, Bar-Ilan
University, Ramat-Gan 52900, Israel, +972 3 531-8770; {\tt
amir@cs.biu.ac.il}; and Department of Computer Science, Johns Hopkins
University, Baltimore, MD 21218. Partly supported by ISF grant 35/05.}
&
Klim Efremenko \thanks{Department of Computer Science,
Bar-Ilan U., 52900 Ramat-Gan, Israel, (972-3)531-8408, {\tt
klimefrem@gmail.com}.} &
Oren Kapah\thanks{Department of Computer Science,
Bar-Ilan U., 52900 Ramat-Gan, Israel, (972-3)531-8408, {\tt
kapaho@cs.biu.ac.il}.}
\\
{\small Bar-Ilan University} & {\small Bar-Ilan University} & {\small
  Bar-Ilan University}\\
{\small and} \\
{\small Johns Hopkins University} \\
\\
Ely Porat\thanks{ Department of Computer Science, Bar-Ilan
University, 52900 Ramat-Gan, Israel, (972-3)531-7620; {\tt
porately@cs.biu.ac.il}.} &&
Amir Rothschild\thanks{ Department of
Computer Science, Bar-Ilan University, 52900 Ramat-Gan, Israel,
(972-3)531-7620; {\tt amirrot@gmail.com}.}
\\
{\small Bar-Ilan University} && {\small Bar-Ilan University}
\end{tabular}
}

\date{}

\maketitle
\begin{abstract}
This paper presents a new technique for deterministic length reduction.
This technique improves the running time of the algorithm presented
in ~\cite{LR07} for performing fast convolution in sparse data.
While the regular fast convolution of vectors $V_1,V_2$ whose sizes
are $N_1,N_2$ respectively, takes $O(N_1 \log N_2)$ using FFT, using
the new technique for length reduction, the algorithm proposed in
~\cite{LR07} performs the convolution in $O(n_1 \log^3 n_1)$, where
$n_1$ is the number of non-zero values in $V_1$. The algorithm
assumes that $V_1$ is given in advance, and $V_2$ is given in
running time. The novel technique presented in this paper improves the
convolution time to $O(n_1 \log^2 n_1)$ {\sl deterministically}, which
equals the best running time given achieved by a {\sl randomized}
algorithm.

The preprocessing time of the new technique remains the same as the
preprocessing time of~\cite{LR07}, which is $O(n_1^2)$. This assumes
and deals the case where $N_1$ is polynomial in
$n_1$. In the case where $N_1$ is exponential in $n_1$, a
reduction to a polynomial case can be used. In this paper we also
improve the preprocessing time of this reduction from
$O(n_1^4)$ to $O(n_1^3{\rm polylog}(n_1))$.
\end{abstract}

\setlength{\parindent}{0.0in}
\setlength{\parskip}{0.1 in}

\section{Introduction\label{sec:intro}}

The {\em  d-Dimensional point set matching problem} serves as powerful
tools in numerous application domains.
In the d-Dimensional point set matching problem, two sets of points
$T,P \in \N^d$ consisting of $n,m$ points, respectively, are given.
The goal is to determine if there is a rigid transformation under
which all the points in $P$ are covered with points in $T$. Among the
important application domains to which this problem contributes are
model based object recognition, image registration, pharmacophore
identification, and searching in music archives. For an explanation of
the uses of the point-set matching problem in these domains
see~\cite{LR07}.

The point-set matching problem has been studied in the literature in
many variation,  not the least of which in the
algorithms literature. In ~\cite{sc:98Sch} Cardoze and Schulman used a
randomized algorithm to reduce the space size of $T,P$ and then
apply solve the problem in the reduced
space. In ~\cite{CH:02} Cole and Hariharan proposed a solution to
the {\it d-Dimensional Sparse Wildcard Matching}. This is a
generalization of the d-Dimensional point set matching problem where
every point in $\N^d$ is associated with a value. A match is declared
if the values of coinciding points are equal. The Cole and Hariharan
solution consists of two steps. The first step is a {\it Dimension
Reduction} where the inputs $T,P$ are linearized into raw vectors
$T',P'$ of size polynomial in the number of non-zero values. The
second step was a {\it Length Reduction} where each of the raw vectors
$T',P'$ was replaced by $\log n$ short vectors of size $O(n)$ where
$n$ is the number of non-zeros. The idea is that the mapping to the
short vectors preserves the distances in the original vectors, thus
the problem is reduced to a matching problem of short vectors, to
which efficient solutions exist. The problem with the {\em length
reduction} idea is that more then one point can be mapped into the
same location, thus it is no longer clear whether there is indeed a
match in the original vectors. The proposed solution of Cole and
Hariharan was to create a set of $\log n$ pairs of vectors using
$\log n$ hash function rather then a single pair of vectors. Their
scheme reduced the failure probability.

In~\cite{LR07}, the first {\bf deterministic} algorithm for
finding $\log n$ hash functions that reduce the size of the vectors to
$O(n \log n)$ was presented. The algorithm guaranteed that each
non-zero value appears with no collisions in {\em at least} one of the
vectors, thus eliminating the possibility of en error.
The {\em length reduction} idea was used to solve the {\it Sparse
Convolution} problem posed in ~\cite{muthu-open}, where the aim is to
find the convolution vector $W$ of two vectors $V_1,V_2$ whose sizes
are $N_1,N_2$, with $n_1,n_2$ non-zero elements respectively (where
$n_1>n_2)$. It is assumed that the two vectors are not given
explicitly, rather they are given as a set of $(index,value)$
pairs. Using the Fast Fourier Transform (FFT) algorithm, the
convolution can be calculated in running time $O(N_1 \log
N_2)$\cite{CLR-92}. In our context, though, the vectors $V_1, V_2$ are
sparse. The aim of the algorithm is to compute $W$ in time
proportional to the number of non-zero entries in $W$, which may be
significantly smaller than $O(N_1)$. Clearly, this can be easily done
in time $O(n_1n_2)$.

The goal of the length reduction is as follows: Given two vectors
$V_1,V_2$ whose sizes are $N_1,N_2$, with $n_1,n_2$ non-zero
elements respectively (where $n_1>n_2)$, obtain two vectors
$V_1',V_2'$ of size $O(n_1)$ such that all the non-zero in $V_1$ and
in $v_2$ will appear as singletons in $V_1'$ and in $V_2'$
respectively while maintaining the distance property.

The distance property which need to be maintained is defined as
follows: If $V_2'[f(0)]$ is aligned with $V_1'[f(i)]$, then
$V_2'[f(j)]$ will be aligned with $V_1'[f(i+j)]$.

This goal was not reached yet, rather a set of $O(\log n_1)$ vectors
of size $O(n_1 \log n_1)$ where obtained in ~\cite{LR07}, where each
non-zero in the text appears at least once as a singleton in the set
of vectors. This length reduction gave an $O(n_1 \log^3 n_1)$
algorithm for convolution in sparse data. In this paper we go one
step forward and reduce the size of the obtained vectors to
$O(n_1)$. This length reduction technique improves the running time
of the fast convolution presented in ~\cite{LR07} to $O(n_1 \log^2
n_1)$, which is the running time for the randomized algorithm
presented in ~\cite{CH:02}.

\section{Preliminaries and Notations}\label{s:pre}

Throughout this paper, a capital letter (usually $N$) is used to
denote the size of the vector, which is equivalent to the largest
index of a non-zero value, and a small letter (usually $n$) is used
to denote the number of non-zero values. It is assumed that the
vectors are not given explicitly, rather they are given as a set of
$(index,value)$ pairs, for all the non-zero values.

A convolution uses two initial functions, $v_1$ and $v_2$, to
produce a third function $w$. We formally define a discrete
convolution.
\begin{definition}
Let $V_1$ be a function whose domain is $\{ 0,..., N_1-1\}$ and $V_2$
a function whose domain is $\{ 0,..., N_2-1 \}$. We may view $V_1$ and
$V_2$ as arrays of numbers, whose lengths are $N_1$ and $N_2$,
respectively. The {\em discrete convolution of $V_1$ and $V_2$} is the
polynomial multiplication
$$ W[j] = \sum_{i=0}^{N_2-1} V_1[j+i] V_2[i].$$
\end{definition}

In the general case, the convolution can be computed by using the Fast
Fourier Transform (FFT)~\cite{CLR-92}. This can be done in time
$O(N_1\log N_2)$, in a computational model with word size $O(\log
N_2)$. In the sparse case, many values of $V_1$ and $V_2$ are
$0$. Thus, they do not contribute to the convolution value. In our
convention, the number of non-zero values of $V_1 (V_2)$ is $n_1
(n_2)$. Clearly, we can compute the convolution in time
$O(n_1n_2)$. The question posed by Muthukrishnan~\cite{muthu-open} is
whether the convolution can be computed in time $o(n_1n_2)$.

Cole and Hariharan's suggestion was to use {\em length reduction}.
Suppose we can map all the non-zero values into a smaller vector, say
of size $O(n_1 \log n_1)$. Suppose also that this mapping is alignment
preserving in the sense that applying the same transformation on $V_2$
will guarantee that the alignments are preserved. Then we can simply
map the the vectors $V_1$ and $V_2$ into the smaller vectors and then
use FFT for the convolutions on the smaller vectors, achieving time
$O(n_1 \log^2 n_1)$.

The problem is that to-date there is no known mapping with that
alignment preserving property. Cole and Hariharan~\cite{CH:02}
suggested a randomized idea that answers the problem with high
probability. The reason their algorithm is not deterministic is the
following:
In their length reduction phase, several indices of
non-zero values in the original vector may be mapped into the same
index in the reduced size vector. If the index of only one non-zero
value is mapped into an index in the reduced size vector, then this
index is denoted as {\it singleton} and the non-zero value is said
to appear as a {\it singleton}. If more then one non-zero value is
mapped into the same index in the reduced size vector, then this
index is denoted as {\it multiple}. The multiple case is problematic
since we can not be sure of the right alignment. Fortunately, Cole
and Hariharan showed a method whereby in $O(\log n_1)$ tries, the
probability that some index will {\em always} be in a multiple
situation is small. In~\cite{LR07}, a deterministic solution to the
multiple problem was presented. That solution utilized number
theoretic ideas. The new idea of this paper is to improve the
reduction size by using {\em polynomials} to represent the location of
the non-$0$ elements of the given vectors.

\section{The New Length Reduction Technique for the Polynomial
  Case}\label{s:length}

The proposed technique deals with the case that $N_1$ is polynomial
in $n_1$, thus the indices are bounded by $n_1^c$. In the case
where, $N_1$ is exponential in $n_1$, the reduction to a polynomial
case can be used.

The main idea of the algorithm is to derive a set of unique polynomials
from each non-zero index in $V_1$, and one polynomial for each non-zero
in $V_2$. Each assignment for the polynomials in $\F_q$, where $q$ is a
prime number of size $\Theta (n_1)$ will give a different mapping of
the non-zeros in $V_1$ and in $V_2$ to vectors of size $q$. The
convolution will be performed between the vectors obtained from
$V_1$ and $V_2$ under the same assignments.

The first step of the algorithm is to choose a prime number of size
$\Theta (n_1)$, and create a polynomial for each non-zero index in
$V_1$. The created polynomial of index $i$ will be denoted as the base
polynomial of $T[i]$. The creation of the polynomial is done by
representing the index as a number in base $(q-1) \over 2$. Each
digit is interpreted as a coefficient of the polynomial. For example:
If $q=13$, then index $95$ in base $10$ is $235$ in base ${(13-1)
\over 2}=6$ which is represented by the polynomial $2X^2+3X+5$.

Since the indices in $V_1$ are bounded by $n_1^c$, and $q$ is $\Theta
(n_1)$, then the degree of the polynomials which created in this step
is bounded by $c$. In the next step, from each polynomial we create
$2^c$ polynomials. This is done by giving to choices for each
coefficient of the polynomial: (1) Leave it as is. (2) Add $(q-1) \over
2$ to the coefficient and decrease by 1 the coefficient of the
higher degree. We do this for all the coefficients of the polynomial
except for the coefficient of the highest degree.

\begin{example}
Suppose we have a non-zero index $95$, using $q=13$ we get the base
polynomial $2X^2+3X+5$. After the second step we will obtain $4$
polynomials: $2X^2+3X+5$, $2X^2+2X+11$,$X^2+9X+5$,$X^2+8X+11$.\\
The first polynomial is the base polynomial. The second polynomial was
obtained by adding $6$ to the first coefficient and decreasing the
second coefficient by one. The 3rd and the 4th polynomials were created
by adding $6$ to the second coefficient of the first and second
polynomials respectively, and decreasing the third coefficient by one.
\end{example}

The duplication of the polynomials was made to meet the distance
preserving requirement from the length reduction specified in the
following Lemma:

\begin{lemma}\label{l:PolAllignment}
For any assignment of $X$, if $V_2[0]$ is aligned with the base
polynomial representing $V_1[i]$, then $V_2[j]$ will be aligned with
one of the polynomials representing $V_1[i+j]$.
\end{lemma}
\begin{proof}
Let $q$ be the chosen prime number. Index $0$ in $V_2$ is
represented by the polynomial $0$, and index $j$ in $V_2$ is
represented by the a polynomial $A=a_cX^c+a_{c-1}X^{c-1}+...+a_0$.
Index $i$ in $V_1$ is represented by a polynomial of the form
$B=b_cX^c+b_{c-1}X^{c-1}+...+b_0$, and index $i+j$ in $V_1$ is
represented by a polynomial $D=d_cX^c+d_{c-1}X^{c-1}+...+d_0$. Note
that the coefficients $a_i$ and $b_i$ are smaller then $(q-1) \over
2$.

Clearly, if $V_2[0]$ is aligned with $V_[i]$, then for any
assignment of $X$, $V_2[j]$ will be aligned with the polynomial
$A+B=(a_c+b_c)X^c+(a_{c-1}+b_{c-1})X^{c-1}+...+(a_0+b_0)$. Now lets
look at the first coefficient of $D$, since $a_0$ and $b_0$ are
smaller then $(q-1) \over 2$, then there are only two cases: (1)
$(a_0+b_0)<{(q-1) \over 2}$, thus $d_0=a_0+b_0$. (2)
$(a_0+b_0)>={(q-1) \over 2}$, thus $d_0=a_0+b_0-{(q-1) \over 2}$
which is covered by the
polynomial where $(q-1) \over 2$ was added to the first coefficient.\\
In the later case, one was added to the second coefficient, thus we
decrease the next coefficient whenever we add $(q-1) \over 2$ to the
current coefficient. The same cases exist also in all the
coefficient, but a polynomial was created for each possible case ($2^c$
cases), thus one of the created polynomials will be equal to the
polynomial $A+B$. $\QED$
\end{proof}

Note that all the $2^c \times n_1$ created polynomials are unique, and
in $\F_q$. Assigning a value to the polynomials in $\F_q$ will give a
vector of size $q$.

\begin{lemma}\label{l:maxRoots}
Any two polynomials can be mapped to the same location in at most
$c$ assignments.
\end{lemma}
\begin{proof}
The distance between any two polynomials gives a polynomial, where the
degree of the difference polynomial is bounded by $c$. Since both
polynomials give the same index under the selected assignment, then the
assigned value is a root of the difference polynomial. The degree of
this polynomial is bounded by $c$, thus it can have at most $c$
different roots in $\F_q$. $\QED$
\end{proof}

Since any polynomial can be mapped into the same location with at most
$2^c \times n_1-1$ other polynomials, and with each of them at most $c$
times, due to Lemma \ref{l:maxRoots}, then we get the following
Corollary:
\begin{corollary}\label{c:maxPolMultiples}
Any polynomial can appear as a {\it multiple} in not more then $c
\times 2^c \times n_1$ vectors.
\end{corollary}

The last step of the length reduction algorithm is to find a set of
$O(\log n_1)$ assignments which will ensure that each polynomial will
appear as a singleton at least once.

The selection of the $O(\log n_1)$ assignments is done as follows:
Construct table $A$ with $2^c \times n_1$ columns and $c \times
2^{c+1} \times n_1$ rows. Row $i$ correspond to an assigned value
$a_i$ and the corresponding reduced length vector $V_{1,i}$. A
column corresponds to a polynomial $P_j$. The value of $A_{ij}$ is set
to $1$ if polynomial $j$ appears as a {\it singleton} in vector
$V_{1,i}$. Due to Corollary \ref{c:maxPolMultiples}, the number of
zeros in each column can not exceed $c \times 2^c \times n_1$. Thus,
in each column there are $1$'s in at least half of the rows, which
means that the table is at least half full. Since the table is at
least half full there exists a row in which there is one in at least
half of the columns. The assignment value which generated this row
is chosen, and all the columns where there was a $1$ in the selected
row are deleted from the table.

Recursively another assignment value is chosen and the table size is
halved again, until all the columns are deleted. Since at each step
at least half of the columns are deleted, the number of prime number
chosen can not exceed  $\log (2^c \times n_1) = c \log n_1$.

{\bf Time:} Creating vector $V_{1,i}$ (row $i$) takes $O(n_1)$ time.
Since we start with a full matrix of $O(n_1)$ rows then the
initialization takes $O(n_1^2)$ time. Choosing the $O(\log n_1)$
assignment values is done recursively. The recurrence is:
$$ t(n_1^2) = n_1^2 + t({n_1^2\over 2})$$
The closed form of this recurrence is $O(n_1^2)$.

\section{The New Algorithm for The Exponential Case}\label{s:Exp}

In this case, as proposed in ~\cite{LR07}, each of the vectors $V_1$
and $V_2$ is reduced into a single vector of size $O(n_1^4)$, where
all the non-zeros appear as singletons. The reduction is preformed
using the modulus function with a prime number $q$ of size
$O(n_1^4)$. It was already proven there that there are at most
$n_1^3$ prime number of size $O(n_1^4)$, which generate at least one
multiple. Thus, by testing $n_1^3+1$ prime numbers we ensure that at
least one of them produce a vector with no multiples.

In order to find such a prime number, we find $n_1^3+1$ of size
$O(n_1^4)$. Then we multiply all the prime numbers to receive a
large number $Q$. In addition we have at most $n_1^2$ different
distances between any two non-zeros. We multiply all of them to
receive the large number $D$. The next step is to find the greatest
common divider ($GCD$) between $Q$ and $D$. Since there is at least
one prime number in $Q$ which does not divide $D$, then $GCD(Q,D)$
is less then $Q$. Dividing $Q$ by the $GCD(Q,D)$ will give $P$ which
is the multiplication of all the prime numbers that create only
singletons. The last step is to find at least one of them. This is
done using a binary search on the prime numbers. We take the
multiplication of half of the prime numbers $Q'$, and find the
$GCD(Q',P)$. If $GCD(Q',P) > 1$ we continue with this set of prime
numbers and multiply half of them iteratively. Otherwise, we
continue with the other half of the prime numbers. After $O(\log
n_1)$ iterations we will find one prime number which will generate
only singletons.

The algorithm appears in detail below.

\fbox{
\begin{minipage}{16cm}
{\bf Algorithm  -- $N_1$ is exponential in $n_1$} {\sf
\begin{enumerate}
    \item Find $n_1^3+1$ prime numbers of size $O(n_1^4)$.
    \item Multiply all the prime numbers to obtain $Q$.
    \item Multiply all the difference between any two non-zero indices to obtain $D$.
    \item Set $P={Q \over GCD(Q,D)}$.
    \item Let $S$ be the set of all prime numbers.
    \item While the size of $S$ is larger then $1$ do:
\begin{enumerate}
    \item Let $S'$ be a set of the first half of prime numbers in $S$.
    \item Set $Q'$ to be the multiplication of all the prime numbers in $S'$.
    \item If $GCD(Q',P)>1$ then set $S=S'$, otherwise set $S=S/S'$.
\end{enumerate}
\end{enumerate}
{\bf end Algorithm} }
\end{minipage}
}

{\bf Correctness:} Immediately follows from the discussion.

{\bf Time:} Step 1 is performed in time $O(n_1^3 {\rm polylog}(n_1))$
using the primality testing described in ~\cite{berri:02}. Step 2 is
done by building a binary tree of multiplication where each node
contain the multiplication of the two number in the lower level. This
tree has $O(\log n_1)$ levels. In the leaves there are $n_1^3$ prime
numbers with $\log n_1$ bits, so the total number of bits in each
level is $O(n_1^3 \log n_1)$. A multiplication of two numbers can be
computed in time $O(b \log b \log \log b)$ ~\cite{SS-71}, where $b$ is
the number of bits. Thus each level can be computed in time $O(n_1^3
{\rm polylog} (n_1))$ and the total time for step 2 is $O(n_1^3
{\rm polylog}(n_1))$. step 3 is preformed in the same way, but this time
in the leaves there are $n_1^2$ numbers with $n_1$ bits, thus each
level has $n_1^3$ bits and the time for this step is $O(n_1^3 \log
n_1)$. In step 4 we calculate the $GCD$ of two numbers with $O(n_1^3
\log n_1)$ bits. This can be calculated in time $O(n_1^3
{\rm polylog}(n_1))$ using ~\cite{SZ:02}. The calculation for step 6(b) was
already performed in step 2, and step 6(c) can be calculated in time
$O(n^3 {\rm polylog}(n_1))$, thus the time of step 6 is $O(n_1^3
{\rm polylog}(n_1))$. Following this discussion the total time of this
algorithm is $O(n_1^3 {\rm polylog}(n_1))$.

\section{Conclusion and Open Problems}\label{s:conc}

Improved deterministic algorithms for Length Reduction and Sparse
Convolution where presented in this paper. These can be used as
tools to provide faster algorithms for several well known problems.
The deterministic time achieved for convolving input patterns with a
fixed text is the same as the best known randomized algorithm.

An important problem remains: Can the Length Reduction and Sparse
Convolution problems be solved in real time without the need of the
preprocessing step, or alternately, can the preprocessing time be
reduced from quadratic?

\bibliographystyle{plain}
\small{
\bibliography{PM}
}

\end{document}